\documentclass[10pt, conference, letterpaper]{IEEEtran}
\ifCLASSINFOpdf
\usepackage[pdftex]{graphicx}
\else
\usepackage[dvips]{graphicx}
\fi

\usepackage[pass,]{geometry}

\usepackage{amsmath,amsthm,amssymb,amsfonts}
\usepackage{epsfig}
\usepackage[tight]{subfigure}
\usepackage{graphicx}
\usepackage{cite}
\usepackage{times}
\usepackage{enumerate}
\usepackage{enumitem}
\usepackage{bm}
\usepackage{algorithmic}
\usepackage{algorithm}
\usepackage[utf8]{inputenc}
\usepackage{url}
\usepackage{lipsum}
\usepackage{fancyhdr}
\usepackage{tikz}
\usetikzlibrary{calc}

\newtheorem{theorem}{Theorem}
\newtheorem{lemma}{Lemma}
\newtheorem{definition}{Definition}
\newtheorem{observation}{Observation}

\newtheorem{corollary}{Corollary}

\let\emptyset\varnothing
\IEEEoverridecommandlockouts
\begin{document}

\title{Robust Design of Spectrum-Sharing Networks}

\author{
\IEEEauthorblockN{Qingkai Liang\IEEEauthorrefmark{1},
Hyang-Won Lee\IEEEauthorrefmark{2}, Eytan Modiano\IEEEauthorrefmark{3}}
\IEEEauthorblockA{\IEEEauthorrefmark{1}\IEEEauthorrefmark{3}Laboratory for Information and Decision Systems, MIT, Cambridge, MA\\
\IEEEauthorrefmark{2}Department of Internet and Multimedia Engineering, Konkuk University, Seoul, Korea}
\thanks{This work was supported by NSF Grants CNS-1116209 and AST-1547331. Hyang-Won Lee was supported by the Basic Science Research Program through the National Research Foundation of Korea (NRF) funded by the Ministry of Science, ICT \& Future Planning (2015R1A1A1A05001477).}
}



\maketitle

\begin{tikzpicture}[remember picture, overlay]
\node at ($(current page.north) + (-3in,-0.5in)$) {Technical Report};
\end{tikzpicture}

\begin{abstract}
In spectrum-sharing networks, primary users have the right to preempt secondary users, which significantly degrades the performance of underlying secondary users. In this paper, we use backup channels to provide reliability guarantees for secondary users. In particular, we study the optimal white channel assignment that minimizes the amount of recovery capacity (i.e., bandwidth of backup channels) needed to meet a given reliability guarantee.  This problem is shown to be coupled by two NP-hard objectives. We characterize the structure of the optimal assignment and develop bi-criteria approximation algorithms. Moreover, we investigate the scaling of the recovery capacity as the network size becomes large. It is shown that the recovery capacity is  negligible as compared to the total traffic demands in a large-scale network.
\end{abstract}

\section{Introduction}\label{sec:introduction}
In spectrum-sharing networks\footnote{The spectrum-sharing feature may be enabled via cognitive radios \cite{Mitola-CR}, geographic databases \cite{white-database}, etc.},  secondary users can access spectrum holes (referred to as white channels) that are not used by primary users. While spectrum sharing enables efficient utilization of spectrum resources, secondary networks built upon white channels can suffer from severe performance degradation since secondary users must stop using a white channel whenever it is reclaimed by a primary user (this event is called \emph{channel preemption}). 
Thus, it is necessary to provide protection for secondary users to guarantee their reliability against channel preemptions.

There have been numerous efforts towards achieving reliable communications for secondary users.  One of the important issues in this context is how the secondary network should recover from channel preemptions. A straightforward approach is to let disrupted links switch to another idle white channel on the fly \cite{Shin-DySPAN-08}\cite{Liang-JSAC-12}. This approach can, however, experience unpredictable delay until idle white channels become available. In contrast to the on-the-fly reconfiguration method, Yue \emph{et al.}\cite{Yue-ICC-08} propose to assign an extra white channel to each link in advance, in order to recover from any single channel preemption. In multi-hop networks, rerouting can be used to find a detour around interrupted links \cite{Chowdhury-JSAC-08}\cite{Chowdhury-adhoc-09}. Some recent works \cite{Liang-TON-14}\cite{Chung-ICC-13} combine channel switching and rerouting to recover secondary users' traffic.  Another line of research focuses on ``risk mitigation", which seeks to reduce the negative effects of channel preemptions on secondary networks. Zhao \emph{et al.} \cite{Zhao-JSAC-07} and Kuo \emph{et al.} \cite{Kuo-WCNC-07} exploit channel statistics to predict channel availability and design reliable MAC protocols to reduce the probability of being preempted. Cao \emph{et al.} \cite{Cao-INFOCOM-12} and Mihnea \emph{et al.} \cite{Mihnea-ICNC-14} study reliable channel assignments that maintain the network connectivity after any single channel preemption. 

Although the above schemes enhance the reliability of secondary networks, most of them only provide ``best-effort reliability". There is no guarantee on, for example, the number of channel preemptions the secondary network can recover from, or the ability to fulfill a certain reliability requirement. In this paper, we allow secondary users to specify a reliability requirement and investigate how to adhere to such a requirement at the minimum cost.

Our approach uses backup channels to recover from preemptions. These backup channels can be licensed channels leased temporarily at a cost \cite{Modiano-Mobihoc-12}, or currently unused white channels. Note that these backup channels do not necessarily stay idle when they are not used for recovery; the only requirement is that they should be available when needed for recovery (possibly at a cost).


Due to the scarcity and relative high costs of backup channels, it is necessary to minimize the amount of \emph{recovery capacity} (i.e., bandwidth of backup channels) that should be provisioned. Although many factors  can affect the amount of required recovery capacity, we focus on the influence of \emph{white channel assignment}. Specifically, we study the optimal white channel assignment that minimizes the recovery capacity required to meet a certain reliability requirement such that the network is able to recover secondary users' traffic from a given number of white channel preemptions.



Unfortunately, this problem is shown to be intractable and coupled by two NP-hard objectives. As a result, we conduct bi-criteria analysis and propose bi-criteria approximation algorithms for white channel assignment. Our simulations validate the performance of the proposed algorithms.

Another important contribution of this paper is the characterization of the scaling of the recovery capacity. It turns out that the required recovery capacity becomes negligible as compared to the total network traffic as the network becomes large. Our simulations show that under the proposed channel assignment schemes the required recovery capacity is usually less than 1\% of the total traffic. Thus, it is possible to provision guaranteed reliability in a large-scale secondary network at minimum cost.

The remainder of this paper is organized as follows. We introduce the network model and describe the problem in Section \ref{sec::model}. Next, we study the optimal white channel assignment under a given reliability requirement in Section \ref{sec::detrec}. Finally, simulation results are presented in Section \ref{sec::simulation} and conclusions are given in Section \ref{sec::conclusion}.


\section{Model and Problem Description}\label{sec::model}

\subsection{Network Model}\label{sec::network-model}

We consider a spectrum-sharing network where primary users own a set of licensed channels referred to as \emph{white channels}. Any idle white channel can be accessed by secondary users, but it should be vacated if a primary user appears in that channel (referred to as \emph{channel preemption}). When channel preemptions happen, secondary users switch to backup channels in order to resume communications.
The \emph{recovery capacity} refers to the bandwidth of backup channels we need to provision in order to meet a certain reliability requirement which will be specified in Section \ref{sec::detrec}.


The secondary network is represented by an undirected graph $G=(V,E)$, where $V$ is the set of secondary nodes and $E$ is the set of links. There is a link between two secondary nodes if they can directly communicate with each other. We consider the one-hop interference model where adjacent links cannot be active in the same channel at the same time. Although such an interference model is restrictive, it serves as the foundation for understanding more complex interference models (e.g., see \cite{one-hop1}\cite{one-hop2}). Moreover, the one-hop interference model is an appropriate model for many practical wireless systems such as spread-spectrum systems, millimeter-wave networks \cite{mmWave}, etc. Each link $e$ is associated with a traffic demand $r_e$ which is determined by some higher-layer policies (e.g., routing and flow control). We denote by $W$ the set of white channels. Each white channel $w$ can sustain a data rate up to $R_{w,e}$ over link $e$.

Now we describe the set of feasibility conditions on white channel assignment in order to sustain the given traffic demands. Let $y$ be an $|E|\times |W|$ binary matrix whose element $y_e^w=1$ if white channel $w$ is assigned to link $e$.  Note that if white channel $w$ is assigned to link $e$, this link should be scheduled for at least $\frac{r_e}{R_{w,e}}$ fraction of time in order to meet the traffic demand $r_e$.  Under the one-hop interference model, the set of links that can be activated simultaneously in the same channel form a \emph{matching}, and interfering matchings can access the same white channel in a time-sharing manner. As a result, the sustainable rate region in each channel can be described by the \emph{convex hull} of all matchings, i.e., the \emph{matching polytope}. Denote by $\mathcal{R}^w$ the rate region that can be sustained in white channel $w$.   Note that under channel assignment $y$, the rate vector that channel $w$ should support is $(y_e^w r_e)_{e\in E}$, which must belong to the rate region $\mathcal{R}^w$ if $y$ is a feasible assignment. Therefore, the feasibility condition can be interpreted as $(y_e^w r_e)_{e\in E}\in \mathcal{R}^w$ for all white channel $w$. Based on Edmond's matching polytope description \cite{Edmond}, we can further write $\mathcal{R}^w$ in a closed form and obtain the following feasibility conditions:

\vspace{-2mm}
\begin{small}
\begin{align}
&\sum_{e\in\delta(v)}\frac{r_e}{R_{w,e}}y_e^w\leq 1,\forall v\in V, w\in W\label{eqn:node-constraint-y}\\
&\sum_{e\in E(U)}\frac{r_e}{R_{w,e}}y_e^w\leq \frac{|U|-1}{2},\forall U\in \mathcal{V}, w\in W\label{eqn:oddset-constraint-y}\\
&\sum_{w\in W}y_e^w=1,\forall e\in E\label{eqn:one-channel-constraint}\\
&y_e^w\in\{0,1\},\forall e\in E, w\in W\nonumber.
\end{align}
\end{small}\vspace{-5mm}

\noindent In \eqref{eqn:node-constraint-y}, we denote by $\delta(v)$ the set of links incident on node $v$. In \eqref{eqn:oddset-constraint-y}, we define $\mathcal{V}=\{U\subseteq V: |U| \text{ odd }\geq3\}$ to be a collection of node sets with odd cardinality, and $E(U)$ is the set of links whose both ends are in $U$. For each white channel $w$, the corresponding constraints in (\ref{eqn:node-constraint-y}) and (\ref{eqn:oddset-constraint-y}) are Edmond's matching polytope description over the set of links using that channel. Specifically, the constraints in \eqref{eqn:node-constraint-y} require the total schedule length of channel $w$ not to exceed one; the constraints in \eqref{eqn:oddset-constraint-y} are called ``odd-set constraints" and we refer readers to \cite{Hajek-IT-88} or \cite{Edmond} for a detailed explanation. Overall, the constraints in (\ref{eqn:node-constraint-y}) and (\ref{eqn:oddset-constraint-y}) force all of the traffic demands to be schedulable under one-hop interference by using the given set of white channels. Hajek \emph{et al.} \cite{Hajek-IT-88} use a similar formulation to characterize schedulability in a single-channel case. Finally, the constraints in (\ref{eqn:one-channel-constraint}) force each link to be assigned exactly one white channel.

A channel assignment $y$ is said to be \emph{feasible} if it satisfies all of the above constraints.
\subsection{Problem Description}\label{sec::description}
Due to the scarcity and relatively high costs of backup channels, it is necessary to minimize the amount of recovery capacity (i.e., bandwidth of backup channels) needed to comply with a certain reliability requirement. In this paper, the secondary network is required to survive a given number of channel preemptions even in the worst case.

\begin{figure}[]
\centering \subfigure[Channel Assignment I]{\label{fig:chan-assn-1}\epsfig{file=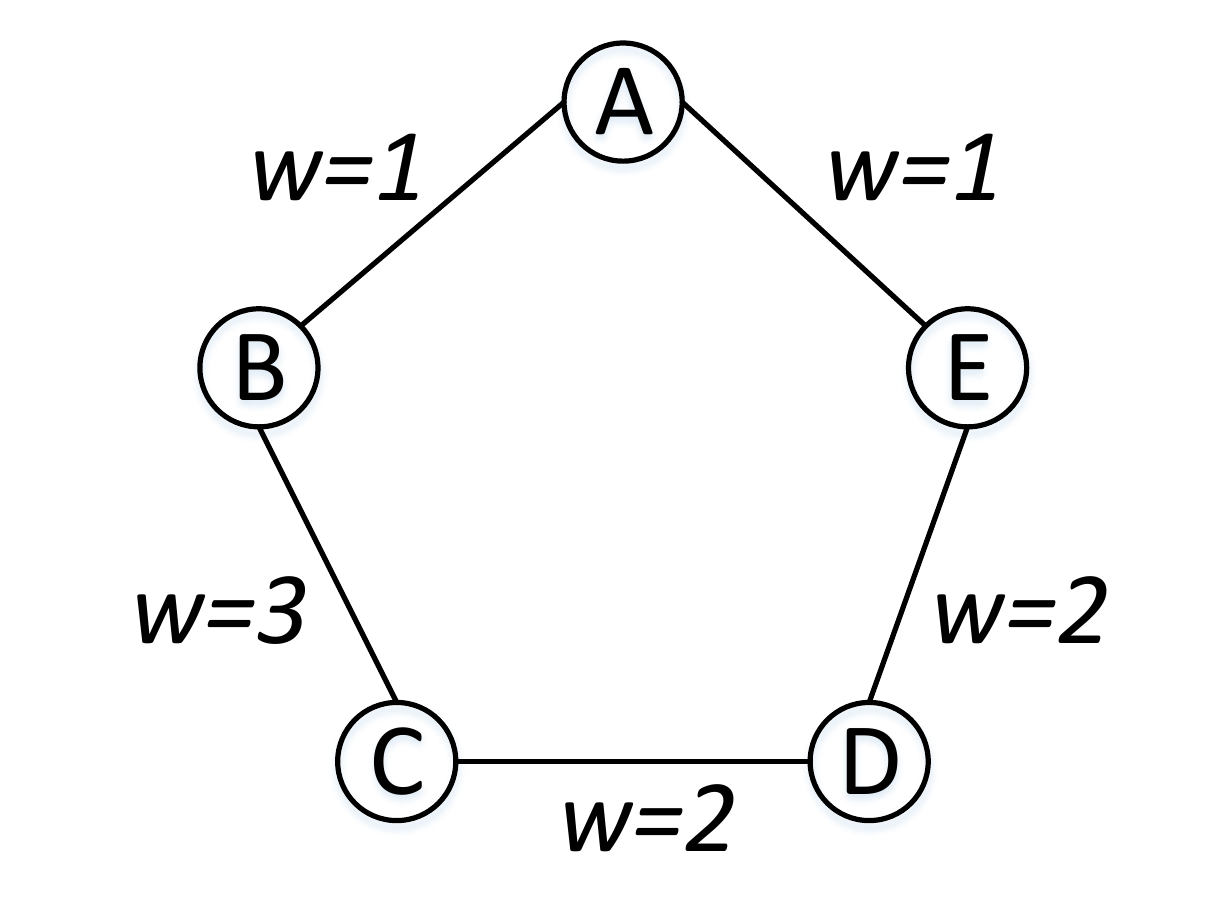,angle=0,width=0.2\textwidth}}\hspace{0.15cm}
\subfigure[Channel Assignment II]{\label{fig:chan-assn-2}\epsfig{file=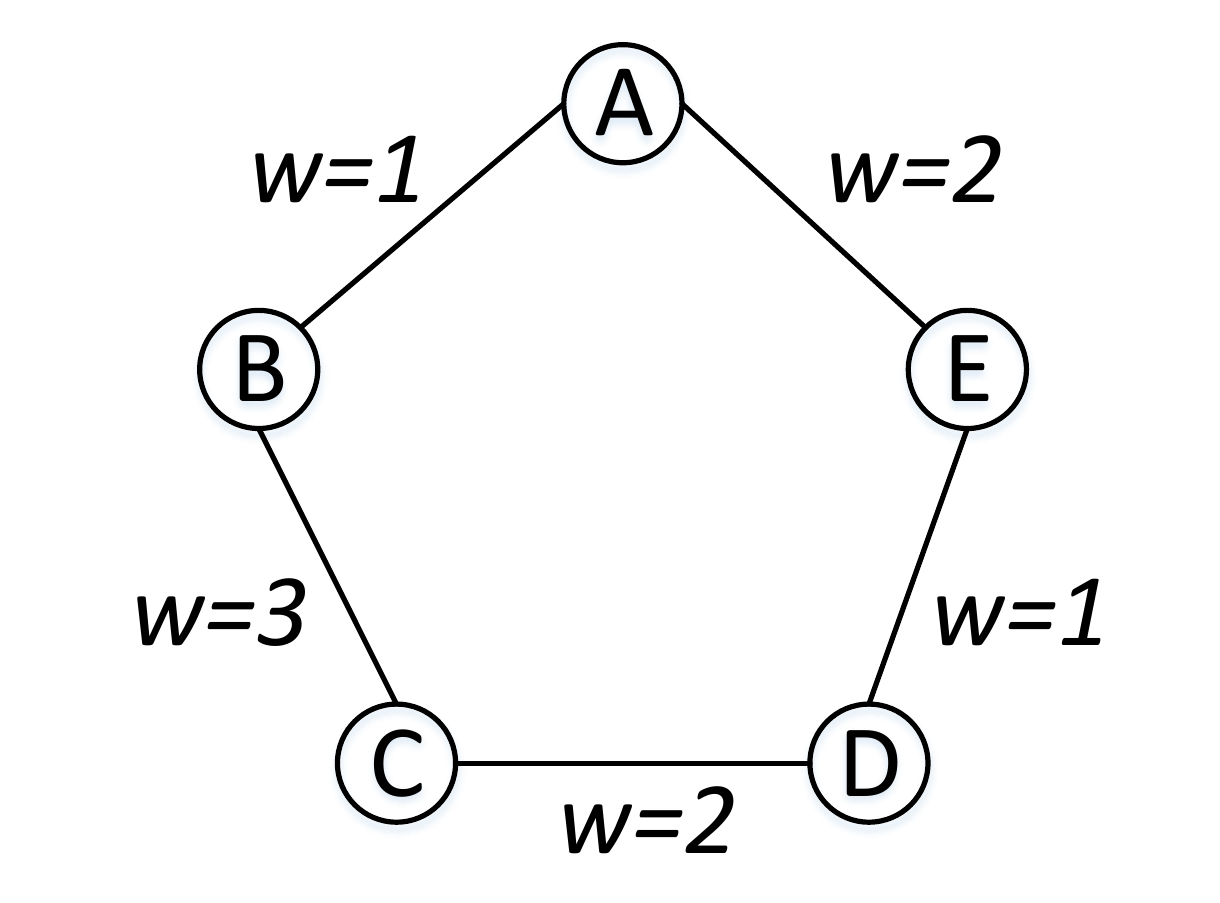,angle=0,width=0.2\textwidth}}\hspace{0.15cm}
\caption{Two different assignments of white channels. }
\label{fig:channel-assignment}\vspace{-5mm}
\end{figure}

Given a recovery requirement, the amount of recovery capacity we need to provision depends on how much traffic is lost due to channel preemptions, which is largely determined by the assignment of white channels. For example, Fig. \ref{fig:channel-assignment} illustrates two different channel assignments with 3 white channels. Each link has one-unit traffic demand, and we assume white channels have sufficiently large capacity such that any channel assignment is feasible (as long as each link is assigned exactly one white channel). Suppose we want to survive any single white channel preemption. In Fig. \ref{fig:chan-assn-1}, the preemption of channel 1 will cause the failures of two adjacent links, which requires two units of recovery capacity under one-hop interference. In contrast, the channel assignment in Fig. \ref{fig:chan-assn-2} only requires one unit of recovery capacity,  since any links that can fail at the same time (i.e., due to the failure of a single white channel) can be activated simultaneously.
\textbf{Our goal is to find a feasible white channel assignment that requires the minimum recovery capacity subject to a certain recovery requirement}.

\section{Robust White Channel Assignment}\label{sec::detrec}
In this section, we develop robust white channel assignment schemes that fulfill a given recovery requirement at minimum cost. Specifically, the network is required to survive any $k$ white channel preemptions, i.e., the backup channels should be able to support the traffic demands on the links disrupted by any $k$ white channel preemptions. Hence, the goal is to find a feasible white channel assignment requiring the minimum recovery capacity to protect against any $k$ channel preemptions. This problem is referred to as $\mathsf{WhiteRec}$:

\vspace{-3mm}
\begin{small}
\begin{align}
&\min_{C,y\text{ feasible}} \quad C\nonumber\\
\textrm{ s.t.} &\sum_{w\in S}\sum_{e\in\delta(v)}\frac{r_e}{C}y_e^w\leq 1,\forall v\in V,\,\,\forall S\in\mathcal{W}(k)\label{eqn:node-constraint-det}\\
&\sum_{w\in S}\sum_{e\in E(U)}\frac{r_e}{C}y_e^w\leq\frac{|U|-1}{2},\forall U\in\mathcal{V}, S\in\mathcal{W}(k)\label{eqn:oddset-constraint-det},
\end{align}
\end{small}\vspace{-2mm}

\noindent where the meanings of $\delta(v)$, $E(U)$ and $\mathcal{V}$ are the same as in \eqref{eqn:node-constraint-y}-\eqref{eqn:oddset-constraint-y}, and $\mathcal{W}(k)=\{S\subseteq W: |S|=k\}$ is a collection of channel sets  with cardinality $k$. Similar to (\ref{eqn:node-constraint-y}) and (\ref{eqn:oddset-constraint-y}), the constraints in (\ref{eqn:node-constraint-det}) and (\ref{eqn:oddset-constraint-det}) correspond to Edmond's matching polytope description, requiring that after any $k$ channel preemptions the traffic demands on the disrupted links  be schedulable by using a backup channel with capacity $C$. 

\subsection{Complexity Analysis}\label{sec::complexity}
In this section, we investigate the complexity of $\mathsf{WhiteRec}$. In fact, solving this problem involves finding a white channel assignment that is both \emph{feasible} (in order to support the traffic demands as described in \eqref{eqn:node-constraint-y}-\eqref{eqn:one-channel-constraint}) and \emph{optimal} (in order to minimize the recovery capacity as defined in \eqref{eqn:node-constraint-det}-\eqref{eqn:oddset-constraint-det}). Unfortunately, both of these problems are NP-hard.\vspace{-1mm}
\begin{theorem}\label{thm::NP-hard1}
Finding a feasible white channel assignment that sustains the given traffic demands is NP-hard.
\end{theorem}
\begin{proof}
Our proof is based on the reduction from the Bin Packing Problem which is known to be NP-hard.
\begin{itemize}
\item Problem: Bin Packing Problem
\item Input: a set of $n$ items with volume $v_1,v_2,...,v_n\in (0,1]$ and a set of $m$ bins with unit capacity
\item Decision: whether we can pack the $n$ items into the $m$ bins
\end{itemize}
To show the reduction, consider a star network with $n$ links incident on a common node. The traffic demands on these links are $v_1,v_2,\cdots,v_n$. Suppose we have $m$ homogeneous white channels, each with unit capacity. We would like to find a mapping from the $n$ links to the $m$ white channels such that the traffic demands are sustainable. Obviously, this is equivalent to determining the feasibility of packing the $n$ items into the  $m$ bins. As a result, it is NP-hard to find a feasible assignment to support the given demands, even in a star network and when channels are homogeneous.
\end{proof}

\begin{theorem}\label{thm::NP-hard2}\vspace{-1mm}
Finding a white channel assignment that requires the minimum recovery capacity is NP-hard. Moreover, \textbf{even if any channel assignment is feasible} (i.e., the capacity of each white channel is sufficiently large such that the traffic demands are always sustainable under any channel assignment), the problem remains NP-hard.
\end{theorem}
\begin{proof}
Our proof is based on the reduction from the Partition Problem which is known to be NP-hard \cite[p. 223]{garey:johnson}.
\begin{itemize}
\item Problem: Partition Problem
\item Input: A set $A$ of positive integers given by $A=\{r_1,...,r_n\}$
\item Output: A subset $S\subset A$ s.t. $\max\{sum(S),sum(A\setminus S)\}$ is minimized, where $sum(T)$ is the sum of all the elements in $T$
\end{itemize}
To show a mapping from the Partition Problem to our problem, consider a star network with $n$ links, where link $i$ has a traffic demand $r_i$. Suppose that we have two white channels $w_1,w_2$, and both channels have sufficiently large capacity such that any channel assignment is feasible. The goal is to find a white channel assignment requiring the minimum recovery channel capacity to recover from any single channel preemption. Hence, it is desirable to balance the loads on each white channel.

More formally, let $E(w)$ be the set of links using white channel $w$. It is easy to see that the recovery channel capacity in the formulation of $\mathsf{WhiteRec}$ can be expressed as
\[
\textstyle{C=\max\left\{\sum_{e\in E(w_1)}r_e,\sum_{e\in E(w_2)}r_e\right\}}.
\]
In this setting, finding a white channel assignment minimizing $C$ is equivalent to finding a subset $S$ in the Partition Problem. This completes the proof.
\end{proof}

\noindent The two theorems imply that $\mathsf{WhiteRec}$ is a complicated problem coupled by two NP-hard objectives: finding a feasible assignment to support the traffic demand and finding an optimal assignment that requires the minimum recovery capacity. To address this difficulty, we introduce a technique called \emph{bi-criteria approximation} \cite{bi-criteria} which allows the feasibility constraints to be violated by a bounded amount while ensuring some approximation ratio with respect to the recovery capacity. The formal definition is as follows.
\begin{definition}[Bi-Criteria Approximation]
An algorithm achieves $(\rho,\phi)$-approximation to $\mathsf{WhiteRec}$ if the following two conditions are satisfied simultaneously.

\noindent - It requires at most $\rho$ times of the minimum recovery capacity.

\noindent - It guarantees that at least $\phi$-fraction of the traffic demand is sustained over each link.
\end{definition}

\noindent  In the following sections, we first analyze the bi-criteria structure of $\mathsf{WhiteRec}$. Based on the analysis,  several approximation algorithms are developed and their bi-criteria approximation ratios are studied.

\subsection{Bi-Criteria Analysis}\label{sec::optimal-structure}
In this section, we investigate the bi-criteria structure of the optimal feasible solution to $\mathsf{WhiteRec}$. Specifically, we are interested in the structure that requires the minimum recovery capacity (i.e., optimality analysis, Sec. \ref{subsec::optimal}) and that sustains the given traffic demands (i.e., feasibility analysis, Sec. \ref{subsec::feasibility}). Finally, the relationship between optimality and feasibility is discussed.
\subsubsection{Optimality Analysis}\label{subsec::optimal}
We first study the structure of the optimal assignment that requires the minimum recovery capacity. The particular form of $\mathsf{WhiteRec}$ allows us to express the required recovery capacity $C$ in a closed form. It is easy to see that constraints in \eqref{eqn:node-constraint-det} are equivalent to
\begin{equation}\label{eqn::M1}
C\ge M_1(y,k),
\end{equation}
where \vspace{-6mm}
\[
\quad \quad M_1(y,k)=\max_{v\in V,S\in\mathcal{W}(k)}\sum_{w\in S}\sum_{e\in\delta(v)}r_ey_e^w.
\]
Similarly, constraints in \eqref{eqn:oddset-constraint-det} are equivalent to
\begin{equation}\label{eqn::M2}
C\ge M_2(y,k),
\end{equation}
where \vspace{-6.8mm}
\[
\quad \quad M_2(y,k)=\max_{U\in\mathcal{V},S\in\mathcal{W}(k)}\frac{2}{|U|-1}\sum_{w\in S}\sum_{e\in E(U)}r_ey_e^w.
\]
Combining \eqref{eqn::M1} and \eqref{eqn::M2}, we can rewrite constraints \eqref{eqn:node-constraint-det} and \eqref{eqn:oddset-constraint-det} in $\mathsf{WhiteRec}$ as
\begin{equation}\label{eqn::M1M2}
C\ge \max\{M_1(y,k),M_2(y,k)\}\triangleq C(y,k).
\end{equation}
In other words, given a white channel assignment $y$, the value of $C(y,k)$ is the minimum recovery capacity required to recover from any $k$ channel preemptions. As a result, $\mathsf{WhiteRec}$ can be rewritten as
\begin{align}
\min_{y}& \quad C(y,k)\nonumber\\
\textrm{ s.t.}& \quad y\text{ is feasible}\nonumber.
\end{align}
Note that $M_2(y,k)$ corresponds to the ``odd-set constraints" in \eqref{eqn:oddset-constraint-det} which are difficult to handle in general. Hence, it is natural to consider the relaxation of $\mathsf{WhiteRec}$ by neglecting $M_2(y,k)$. The relaxed problem is referred to as $\mathsf{WhiteRecApprox}$, i.e.,
\begin{align}
\min_{y}& \quad M_1(y,k)\nonumber\\
\textrm{ s.t.}& \quad y\text{ is feasible}\nonumber.
\end{align}

\noindent The following lemma shows that the relaxation of $M_2(y,k)$ only leads to a small loss in optimality.

\begin{lemma}\label{thm:C(y,k)-bounds}
For any channel assignment $y$, we have
\begin{equation}\label{eqn:approx-rate-C-M1}
M_1(y,k)\leq C(y,k)\leq1.5M_1(y,k).
\end{equation}
\end{lemma}
\begin{proof}
The lower bound follows from the definition of $C(y,k)$.  To show the upper bound, we notice that

\vspace{-3mm}
\begin{small}
\begin{align*}
\vspace{-1mm}
M_2(y,k)&=\max_{U\in\mathcal{V},S\in\mathcal{W}(k)}\frac{2}{|U|-1}\sum_{w\in S}\sum_{e\in E(U)}r_ey_e^w\\
&=\frac{1}{2}\max_{U\in\mathcal{V},S\in\mathcal{W}(k)}\frac{2}{|U|-1}\sum_{w\in S}\sum_{v\in U}\sum_{e\in \delta(v)\cap E(U)}r_ey_e^w\\
&\leq\frac{1}{2}\max_{U\in\mathcal{V},S\in\mathcal{W}(k)}\frac{2}{|U|-1}\sum_{v\in U}\sum_{w\in S}\sum_{e\in \delta(v)}r_ey_e^w\\
&\leq\frac{1}{2}\max_{U\in\mathcal{V},S\in\mathcal{W}(k)}\frac{2}{|U|-1}\sum_{v\in U}M_1(y,k)\\
& = \frac{1}{2} M_1(y,k)\max_{U\in\mathcal{V}}\frac{2|U|}{|U|-1}=\frac{3}{2}M_1(y,k).
\vspace{-1mm}
\end{align*}
\end{small}\vspace{-3mm}

\noindent The second inequality is due to the definition of $M_1(y,k)$, and the last equality holds because $|U|\ge 3$.
\end{proof}

\noindent  Lemma \ref{thm:C(y,k)-bounds} shows that the optimal solution to $\mathsf{WhiteRecApprox}$ attains 1.5-approximation to the original problem $\mathsf{WhiteRec}$ with respect to the required recovery capacity $C(y,k)$.
In fact, in bipartite graphs, there is even no approximation gap between $\mathsf{WhiteRecApprox}$ and $\mathsf{WhiteRec}$.

\vspace{1mm}

\noindent \textbf{Special Case: Bipartite Graphs.} The notion of bipartite graphs can characterize any graph without odd-length cycles such as trees. Using the particular structures of bipartite graphs, we can show Lemma \ref{thm::bipartite-C(y,k)-bound}.

\begin{lemma}\label{thm::bipartite-C(y,k)-bound}
In a bipartite network, $M_1(y,k)\ge M_2(y,k)$ for any channel assignment $y$ and any positive integer $k$.
\end{lemma}
\begin{proof}
See Appendix \ref{proof::thm::bipartite-C(y,k)-bound}.
\end{proof}

\noindent This lemma implies that $C(y,k)=M_1(y,k)$ in a bipartite graph, and thus we can safely ignore $M_2(y,k)$ without sacrificing any optimality. In other words, the optimal solution to the relaxed problem $\mathsf{WhiteRecApprox}$ is also the optimal solution to the original problem $\mathsf{WhiteRec}$ if the secondary network is bipartite.


\subsubsection{Feasibility  Analysis}\label{subsec::feasibility}
Next, we study the feasibility conditions \eqref{eqn:node-constraint-y}-\eqref{eqn:one-channel-constraint} and investigate the structure of channel assignments that are able to sustain the largest amount of traffic. In particular, we investigate the relationship between feasibility and optimality, which is important for our subsequent bi-criteria approximation analysis.

It is clear  that deciding feasibility is equivalent to the following optimization problem $\mathsf{FEASI}$ which finds the maximum fraction of traffic that can be sustained over each link.

\vspace{-4mm}
\begin{small}
\begin{align}
\mathsf{FEASI}:\quad &\max_{y,\beta} \quad \beta\nonumber\\
\textrm{ s.t.}&\sum_{e\in\delta(v)} \frac {\beta r_e}{R_{w,e}} y_e^w\le 1,\forall v\in V, w\in W \label{AUX_c1}\\
&\sum_{e\in E(U)}\frac{\beta r_e}{R_{w,e}} y_e^w\le \frac{|U|-1}{2},\forall U\in\mathcal{V}, w\in W\label{AUX_c2}\\
& \sum_{w\in W}y_e^w=1,\forall e\in E \label{AUX_c3}\\
&y_e^w\in\{0,1\},\forall e\in E, w\in W\nonumber.
\end{align}
\end{small}\vspace{-4mm}

\noindent Clearly, the original problem $\mathsf{WhiteRec}$ is feasible if and only if the optimal value $\beta^*$ in $\mathsf{FEASI}$ is greater or equal to 1. Now let $\beta(y)$ be the maximum value of $\beta$ in $\mathsf{FEASI}$ under an assignment $y$. The following lemma relates $\beta(y)$ to $C(y,1)$.

\begin{lemma}\label{lm::feasibility}
\[
\small
\frac{R_{\min}}{C(y,1)} \le \beta(y) \le \frac{R_{\max}}{C(y,1)},
\]
where $R_{\min}$ and $R_{\max}$ are the minimum and the maximum white channel capacity, respectively.
\end{lemma}
\begin{proof}
From constraints \eqref{AUX_c1}-\eqref{AUX_c2}, it follows that
\[
\small
\beta(y)=\min \{Z_1(y),Z_2(y)\},
\]
where \vspace{-3.5mm}
\[
\small
Z_1(y)=\min_{v\in V, w\in W}\frac{1}{\sum_{e\in\delta(v)}r_e y_e^w/R_{w,e}},
\]
\[
\small
Z_2(y)=\min_{U\in\mathcal{V}, w\in W}\frac{1}{\frac{2\sum_{e\in E(U)} r_e y_e^w/R_{w,e}}{|U|-1}}.
\]
Under the above notations, $\mathsf{FEASI}$ becomes

\vspace{-3mm}
\begin{small}
\begin{align}
\max_{y}& \quad \min\{Z_1(y),Z_2(y)\}\nonumber\\
\textrm{ s.t.}&\sum_{w\in W}y_e^w=1,\forall e\in E\nonumber\\
&y_e^w\in\{0,1\},\forall e\in E, w\in W\nonumber.
\end{align}
\end{small}\vspace{-5mm}

\noindent It is clear that
\[
\small\begin{split}
Z_1(y)\ge \frac{R_{\min}}{\max_{v\in V, w\in W}\sum_{e\in\delta(v)}r_e y_e^w}=\frac{R_{\min}}{M_1(y,1)},
\end{split}
\]
and similarly we have
$
Z_2(y)\ge \frac{R_{\min}}{M_2(y,1)}.
$
Then we obtain that
\[
\small
\begin{split}
\beta(y)&\ge \min\{\frac{R_{\min}}{M_1(y,1)}, \frac{R_{\min}}{M_2(y,1)}\}\\
&=\frac{R_{\min}}{\max\{M_1(y,1),M_2(y,1)\}}=\frac{R_{\min}}{C(y,1)}.
\end{split}
\]
Similarly, we can show
$
\beta(y)\le \frac{R_{\max}}{C(y,1)}.
$
\end{proof}
\noindent This lemma shows that if an assignment $y$ yields a smaller $C(y,1)$, it tends to sustain more traffic. In particular, if white channels are homogeneous with capacity $R$, the lemma implies $\beta(y) =\frac{R}{C(y,1)}$. In this case, minimizing the recovery capacity required to survive a \textbf{single} preemption is equivalent to maximizing the amount of sustainable traffic. Therefore, this lemma bridges feasibility and optimality, which is important for our subsequent bi-criteria approximation analysis.
\subsection{Algorithm 1: Greedy Algorithm}\label{sec::naive-alg}
In this section, we propose a simple greedy algorithm to solve $\mathsf{WhiteRec}$ and analyze its bi-criteria approximation ratio.

Without loss of generality, let the links in the secondary network be indexed by $e_1,\cdots,e_n$, where $n=|E|$. The greedy algorithm assigns a white channel to each of these links sequentially. Suppose we are deciding the channel assignment for link $e_i=(u,v)$, and define $\delta(u,v)=\delta(u)\cup\delta(v)$, i.e., the set of links incident on node $u$ or node $v$. The greedy rule is to pick the white channel that currently sustains the least traffic over the links in $\delta(u,v)$. The detailed procedures are presented in Algorithm \ref{alg::greedy}, where $E_w$ corresponds to the set of links that are assigned channel $w$.

\begin{algorithm}[ht]
\caption{Greedy White Channel Assignment}\label{alg::greedy}
\begin{algorithmic}[1]
\STATE Initialize  $E_w=\emptyset,~\forall w\in W$;
\FOR{$e_i=e_1,\cdots,e_n$}
\STATE Assign white channel $w^*$ to link $e_i=(u,v)$, where $w^*=\arg\min_{w\in W}\sum_{e\in \delta(u,v)\cap E_w}r_e$;\label{s4}
\STATE $E_{w^*}\leftarrow E_{w^*}\cup \{e_i\}$;
\ENDFOR
\end{algorithmic}
\end{algorithm}

The bi-criteria approximation ratio of this greedy algorithm is given in Theorem \ref{thm::greedy-approx}, where we define $R_{\min}$ and $R_{\max}$ to be the minimum and the maximum white channel capacity, respectively.

\begin{theorem}\label{thm::greedy-approx}
Suppose there exists a feasible solution to $\mathsf{WhiteRec}$. Then the greedy algorithm achieves $(\rho, \frac{1}{\rho}\frac{R_{\min}}{R_{\max}})$-approximation to $\mathsf{WhiteRec}$, where $\rho=\frac{3}{2}(3-\frac{2}{|W|})$.
\end{theorem}
\begin{proof}
See Appendix \ref{proof::thm::greedy-approx}.
\end{proof}
\noindent For instance, if there are 2 homogeneous white channels, the greedy algorithm is guaranteed to sustain at least $\frac{1}{3}$ traffic demands while requiring less than 3 times of the minimum recovery capacity in $\mathsf{WhiteRec}$.

The advantage of the greedy algorithm is in its simplicity. In fact, it does not require any global information when assigning channels for each individual link; thus, this greedy algorithm can even be implemented in a distributed manner, where more fresh local information can be used to improve the overall performance. Moreover, it is applicable to arbitrary networks. Although the theoretical approximation ratio of this algorithm is relatively loose, its practical performance turns out to be much better than the theoretical guarantee\footnote{Similar greedy algorithms have been shown to perform extremely well for frequency assignment in WDM-based optical networks \cite{WDM}.} (see Section \ref{sec::simulation}). Moreover, it is possible to improve the approximation ratio in a wide range of graphs. For example, with Lemma \ref{thm::bipartite-C(y,k)-bound}, it can be shown that the approximation ratio can be improved by a factor of 1.5 in bipartite graphs.

\subsection{Algorithm 2: Interference-Free Assignment}\label{sec::improve-alg}
The above greedy algorithm is simple and has provable performance in any scenario but suffers from the relatively loose approximation ratio. In this section, we discuss an alternative channel assignment scheme, called Interference-Free Assignment (IFA), which is less general than the greedy algorithm but achieves much better performance.
\begin{definition}[Interference-Free Assignment]\label{def::IFA}
An  assignment $y$ is said to be interference-free if any two interfering links are assigned distinct white channels.
\end{definition}
\noindent For example, the channel assignment in Fig. \ref{fig:chan-assn-2} is interference-free while the one in Fig. \ref{fig:chan-assn-1} is not. Conceivably, IFA requires less recovery capacity since links that fail together due to any single channel preemption do not interfere with each other and can be activated simultaneously. Through the rest of this section, we study the properties of IFA. In particular, we will show IFA has nearly-optimal performance.

We first investigate the conditions for the existence of IFA. Note that IFA requires that adjacent links be assigned different channels; this is similar to \emph{edge coloring} where each white channel corresponds to a color. From Vizing's Theorem \cite{vizing:chromatic} for edge coloring, we have the following observation:

\begin{observation}\label{thm:existence-ifca}
There exists an interference-free channel assignment if the number of white channels is greater than the maximum node degree, i.e., $|W|>d_{\max}$.
\end{observation}
\noindent The above observation shows that IFA does not always exist and is thus less general than the greedy algorithm. However, the condition shown in  the above observation is very mild in practice since the number of white channels is usually much larger than the number of neighbors a node has \cite{Liang-TON-14}.

Now we develop an algorithm for constructing an interference-free assignment (Algorithm \ref{alg::IFA}). This algorithm gives an interference-free assignment whenever $|W|>d_{\max}$. Note that this algorithm is still valid if  $|W|\le d_{\max}$ but it does not have a provable performance in this case. Note also that this algorithm colors edges with white channels and there are several polynomial-time algorithms that can perform edge-coloring with $d_{\max}+1$ colors in a simple graph (e.g., \cite{coloring-alg}), therefore Algorithm \ref{alg::IFA} can be run in polynomial time.

\begin{algorithm}[ht]
\caption{Interference-Free Channel Assignment}\label{alg::IFA}
\begin{algorithmic}[1]
\STATE Color the graph with $d_{\max}+1$ colors, which partitions the edges into $d_{\max}+1$ matchings; \\ // These matchings are denoted by $I_1,\cdots,I_{d_{\max}+1}$.
\FOR{$i=1:d_{\max}+1$}
\STATE Assign edges in matching $I_i$ to white channel $w_i$, where $w_i=i\mod~|W|$;
\ENDFOR
\end{algorithmic}
\end{algorithm}

Next, we investigate the properties of IFA. The most important one is given in Lemma \ref{lem:ifca-minimize-m1} which shows that \textbf{any} interference-free channel assignment minimizes $M_1(y,k)$.
\begin{lemma}\label{lem:ifca-minimize-m1}
Consider any two interference-free channel assignments $\bar{y},~\tilde{y}$ and any non-interference-free assignment $\hat{y}$. Then the following relationship holds: $M_1(\bar{y},k)=M_1(\tilde{y},k)\le M_1(\hat{y},k)$ for all $k\in\mathbb{Z}^+$.
\end{lemma}
\begin{proof}
For any interference-free assignment $\bar{y}$, let $\bar{S}\in\mathcal{W}(k)$ and $\bar{v}\in V$ be such that
\begin{equation}
M_1(\bar{y},k)=\sum_{w\in \bar{S}}\sum_{e\in\delta(\bar{v})}r_e\bar{y}_e^w.
\end{equation}
Since $\bar{y}$ is interference-free, all the links incident on a node are assigned different white channels. This is also true for another interference-free assignment $\tilde{y}$. Thus, there exists a set $\tilde{S}\in\mathcal{W}(k)$ such that
\[
\{e\in\delta(\bar{v}):\sum_{w\in\tilde{S}}\tilde{y}_e^w=1\}=\{e\in\delta(\bar{v}):\sum_{w\in \bar{S}}\bar{y}^{w}_e=1\}.
\]
Therefore, we have
\[
\sum_{w\in \bar{S}}\sum_{e\in\delta(\bar{v})}r_e\bar{y}_e^w=\sum_{w\in \tilde{S}}\sum_{e\in\delta(\bar{v})}r_e\tilde{y}_e^w,
\]
which implies $M_1(\bar{y},k)\le M_1(\tilde{y},k)$ by the definition of $M_1(y,k)$. Similarly, we can prove $M_1(\bar{y},k)\ge M_1(\tilde{y},k)$. As a result, it follows that $M_1(\bar{y},k)=M_1(\tilde{y},k)$  for any interference-free channel assignments $\bar{y}$ and $\tilde{y}$.

To prove the second part, consider a non-interference-free channel assignment $\hat{y}$. Obviously, under the asignment $\hat{y}$, the preemption of $k$ white channels can possibly lead to the preemption of more than $k$ links incident on a node. Hence, there exists a set $\hat{S}\in\mathcal{W}(k)$ such that
\[
\{e\in\delta(\bar{v}):\sum_{w\in\hat{S}}\hat{y}_e^w=1\}\supseteq\{e\in\delta(\bar{v}):\sum_{w\in \bar{S}}\bar{y}^{w}_e=1\}.
\]
Therefore, we can conclude that $M_1(\bar{y},k)\leq M_1(\hat{y},k)$.
\end{proof}
\noindent Lemma  \ref{lem:ifca-minimize-m1} together with Lemma \ref{thm:C(y,k)-bounds} immediately implies that IFA achieves no more than 1.5 times of the minimum recovery capacity. In fact, we can further tighten this bound, as shown in the following theorem.
\begin{theorem}\label{thm::ifca-bound}
Suppose there is a feasible solution to $\mathsf{WhiteRec}$ and an interference-free assignment exists. Then any interference-free assignment achieves $(\frac{5}{4},\frac{R_{\min}}{R_{\max}})$ approximation to $\mathsf{WhiteRec}$.
\end{theorem}
\begin{proof}
We first prove that any IFA achieves no more than $\frac{5}{4}$ times the minimum recovery capacity. We start by introducing a lemma whose proof is similar to Lemma \ref{thm:C(y,k)-bounds} and thus omitted.
\begin{lemma}\label{im_lm}
Let $\mathcal{V}'=\{U\subseteq V:~|U|\ge 5, |U|~\text{odd}\}$. Then for any assignment $y$ and integer $k\ge 1$:
\small
\[
\max_{U\in \mathcal{V}', S\in\mathcal{W}(k)}\frac{2}{|U|-1} \sum_{w\in S}\sum_{e\in E(U)}r_ey_e^w\le \frac{5}{4}M_1(y,k).
\]
\end{lemma}

Now we get down to proving that any IFA achieves no more than $\frac{5}{4}$ times of the minimum recovery capacity. Denote $\mathcal{V}_3$ the collection of node sets with cardinality 3. For any channel assignment $y$ and any integer $k\ge 1$, we rewrite $C(y,k)$ as:
\[
\small
\begin{split}
C(y,k)&=\max\{M_1(y,k),M_2(y,k)\}\\
&=\max\Big\{M_1(y,k), \max_{U\in \mathcal{V}_3,S\in\mathcal{W}(k)}\frac{2}{|U|-1} \sum_{w\in S}\sum_{e\in E(U)}r_ey_e^w,\\
& \qquad \qquad \qquad \qquad \max_{U\in \mathcal{V}', S\in\mathcal{W}(k)}\frac{2}{|U|-1} \sum_{w\in S}\sum_{e\in E(U)}r_ey_e^w \Big\}\\
& \stackrel{\Delta}{=}\max\{M_1(y,k),A(y,k),B(y,k)\}.
\end{split}
\]
Let $y'$ be an arbitrary IFA and $y^*$ be the optimal solution to $\mathsf{WhiteRec}$. We observe three key facts:
\begin{itemize}
\item[(1)] $A(y',k)\le A(y^*,k)$. This is due to the fact that in any induced graph of 3 nodes the interference-free assignment $y'$ allocates different channels to different edges, which is optimal in that induced graph.
\item[(2)] $B(y',k)\le \frac{5}{4} M_1(y',k)\le \frac{5}{4} M_1(y^*,k)$. This is due to Lemma \ref{im_lm} and the fact that any IFA minimizes $M_1(y,k)$.
\item[(3)] $M_2(y',k)\ge M_1(y',k)$ otherwise $C(y',k)=M_1(y',k)\le M_1(y^*,k)\le C(y^*,k)$, which implies that $y'$ is optimal. This fact further shows that $C(y',k)=\max\{A(y',k),B(y',k)\}$.
\end{itemize}

\noindent Then we have
\[
\small
\begin{split}
\frac{C(y',k)}{C(y^*,k)}&=\frac{\max\{A(y',k),B(y',k)\}}{\max\{M_1(y^*,k),A(y^*,k),B(y^*,k)\}}\\
&\le \max\Big\{\frac{A(y',k)}{A(y^*,k)},\frac{B(y',k)}{M_1(y^*,k)}\Big\}\le \frac{5}{4}.
\end{split}
\]

Next, we prove that any IFA $y'$ can sustain $\frac{R_{\min}}{R_{\max}}$-fraction of traffic over each link. Note that $C(y',1)=M_1(y',1)$ under the IFA $y'$. By Lemma \ref{lem:ifca-minimize-m1}, the IFA $y'$ minimizes $M_1(y,1)$, so we have $C(y',1)=M_1(y',1)\le M_1(y,1)\le C(y,1)$ for any assignment $y$. Let $\hat{y}$ be a feasible solution to $\mathsf{WhiteRec}$, i.e., $\beta(\hat{y})\ge 1$. Then it follows that for any IFA $y'$
\[
\small
\beta(y')\ge \frac{R_{\min}}{C(y',1)}\ge \frac{R_{\min}}{R_{\max}}\frac{R_{\max}}{C(\hat{y},1)}\ge \frac{R_{\min}}{R_{\max}}\beta(\hat{y})\ge \frac{R_{\min}}{R_{\max}},
\]
where the first and third inequalities are due to Lemma \ref{lm::feasibility}, the second inequality is due to our claim that $C(y',1)\le C(\hat{y},1)$ and the last inequality holds because of our assumption that $\beta(\hat{y})\ge 1$. This completes our proof.
\end{proof}
Note that IFA has a much better approximation ratio than the greedy algorithm with respect to both the recovery capacity and the sustainable traffic. In particular, if channels are homogeneous, then any interference-free assignment is guaranteed to sustain 100\% traffic demands while requiring less than 1.25 times of the minimum recovery capacity. The caveat is that such a good approximation ratio only holds true when IFA exists (i.e., when $|W|>d_{\max}$).
In fact, IFA is even optimal with respect to the recovery capacity in many scenarios, as is shown in Corollary \ref{thm::ifca-optimal}.
\begin{corollary}\label{thm::ifca-optimal}
Suppose there is a feasible solution to $\mathsf{WhiteRec}$ and an interference-free assignment exists. Then any interference-free assignment achieves $(1,\frac{R_{\min}}{R_{\max}})$ approximation to $\mathsf{WhiteRec}$ in \textbf{any} of the following scenarios:\\
(i) $k=1$, i.e., we need to survive any single preemption;\\
(ii) $r_e=r$ $\forall e\in E$ and $k\le d_{\max}$, i.e., traffic is uniform and no more than $d_{\max}$ preemptions are to be survived;\\
(iii) the secondary network is bipartite;\\
\end{corollary}
\begin{proof}
See Appendix \ref{proof::thm::ifca-optimal}.
\end{proof}

\noindent The common feature of the above scenarios is that $C(y,k)=M_1(y,k)$ holds for any interference-free assignment $y$; as a result, Lemma \ref{lem:ifca-minimize-m1} implies that any interference-free assignment minimizes $C(y,k)$ in these cases.
Note that if white channels are homogeneous, then any interference-free assignment requires the minimum recovery capacity and 100\% traffic demands can be sustained in any of the above scenarios. In other words, IFA is both feasible and optimal in these cases.
%

\subsection{Scaling of Recovery Capacity}\label{sec::scaling}
In this section, we investigate the scaling of the required recovery capacity under the proposed algorithms. Specifically, we show that the required recovery capacity becomes  negligible as compared to the total traffic if the network size is relatively large.

To facilitate our analysis, we make a simplified assumption that traffic is uniform across the entire secondary network, i.e., $r_e=r$ for any $e\in E$. Also assume that white channels are homogeneous, i.e., $R_{w,e}=R$ for any $w\in W$ and $e\in E$. Denote by $C^*(k)$ the recovery capacity required to protect against any $k$ channel preemptions under Algorithm \ref{alg::IFA}. Note that Algorithm \ref{alg::IFA} is still valid when $|W|\le d_{\max}$, although the resulted assignment may not be interference-free (yet our subsequent analysis on recovery capacity scaling does not require that the assignment be interference-free). Also let $L_{tot}$ be the total traffic demands in the secondary network, i.e., $L_{tot}=\sum_{e\in E}r_e=r|E|$. The following theorem shows the scaling of the relative recovery capacity ratio $C^*(k)\slash L_{tot}$ with the network size $|V|$.

\begin{theorem}\label{thm::capacity-ratio}
Suppose there is a feasible solution to $\mathsf{WhiteRec}$. Then $\frac{C^*(k)}{L_{tot}}=O(\frac{1}{|V|})$ as $|V|\rightarrow\infty$ for any $k\in\mathbb{Z}^+$.
\end{theorem}
\begin{proof}
Consider the channel assignment scheme shown in Algorithm \ref{alg::IFA}. Clearly, each white channel is assigned to at most $\left\lceil\frac{d_{\max}+1}{|W|}\right\rceil$ matchings; thus, at most $\left\lceil\frac{d_{\max}+1}{|W|}\right\rceil$ links incident on the same node are assigned the same channel. Denote $y$ the above white channel assignment. It follows that
\begin{equation}\label{eqn:suff-cond-bound-nc}
\sum_{e\in\delta(v)}r_ey_e^w\leq r \left\lceil\frac{d_{\max}+1}{|W|}\right\rceil,\forall w,v
\end{equation}
Note that the  matching sets derived in Algorithm \ref{alg::IFA} is also a matching set partition of $E(U)$ for each $U\in\mathcal{V}$. Hence, each white channel is assigned to at most $\left\lceil\frac{d_{\max}+1}{|W|}\right\rceil$ matchings in the matching set partition of $E(U)$. Since each matching of $E(U)$ has at most $\frac{|U|-1}{2}$ edges, it follows that for each $w\in W$ and $U\in\mathcal{V}$
\begin{equation}\label{eqn:suff-cond-bound-osc}
\sum_{e\in E(U)}r_ey_e^w\leq r \frac{|U|-1}{2}\left\lceil\frac{d_{\max}+1}{|W|}\right\rceil.
\end{equation}
By \eqref{eqn:suff-cond-bound-nc} and \eqref{eqn:suff-cond-bound-osc}, we can see that $M_1(y,k)$ and $M_2(y,k)$ are upper bounded by
\begin{equation}\label{eqn:capa-ratio-m2-bnd}
rk\left\lceil\frac{d_{\max}+1}{|W|}\right\rceil.
\end{equation}
It follows that
\begin{align*}
C^*(k)&\leq\max\{M_1(y,k),M_2(y,k)\}\\
&\leq rk\left\lceil\frac{d_{\max}+1}{|W|}\right\rceil.
\end{align*}
If $\mathsf{WhiteRec}$ is feasible, then we have for any $v\in V$
\[
\sum_{e\in\delta(v)}\sum_{w\in W} r_ey_e^w=r\sum_{e\in\delta(v)}\sum_{w\in W} y_e^w=r d_v\le |W|R,
\]
where $d_v$ is the degree of node $v$. Then it follows that $d_{\max}\le \frac{|W|R}{r}$, which implies that
\begin{equation}\label{eqn::scale1}
C^*(k)\le rk \left\lceil\frac{R}{r}+\frac{1}{|W|}\right\rceil.
\end{equation}
At the same time, it is easy to see that
\begin{equation}\label{eqn::scale2}
L_{tot}=r|E|\ge r\frac{d_{\min}|V|}{2}\ge \frac{r|V|}{2}.
\end{equation}
Dividing \eqref{eqn::scale1} by \eqref{eqn::scale2}  yields the desired result. Note that $R$, $r$, $k$ and $|W|$ are regarded as asymptotically constant factors when compared to $|V|$.
\end{proof}

\noindent \textbf{Remark.} The proof to Theorem \ref{thm::capacity-ratio} is specific to Algorithm \ref{alg::IFA}. However, since the gap between Algorithm \ref{alg::greedy} and \ref{alg::IFA} (in terms of the ratio between their required recovery capacities) is a constant, we can conclude that the scaling of recovery capacity under Algorithm \ref{alg::greedy} is also $O(\frac{1}{|V|})$ as $|V|\rightarrow\infty$.

\vspace{1mm}

Theorem \ref{thm::capacity-ratio} demonstrates that as the network size grows, the recovery capacity needed to protect against $k$ white channel preemptions becomes negligible as compared to the total traffic. Our simulation results (see Section \ref{sec::simulation-scaling}) show that the recovery capacity required to survive 2 preemptions is less than 1\% of the total traffic in a 200-node network, even with very few white channels. This is mainly due to the effect of spatial reuse. That is, although the total traffic increases linearly with the network size, more links can be activated simultaneously; thus, the required recovery capacity does not scale up with the network size.

\section{Performance Evaluation}\label{sec::simulation}
In this section, we numerically study our schemes. Specifically, we seek to answer the following questions:

\vspace{1mm}

\noindent $\bullet$ How does the recovery capacity scale with the network size?

\vspace{1mm}

\noindent $\bullet$ What is the bi-criteria approximation quality of the greedy algorithm and IFA?

\subsection{Simulation Setup}
We use Erd\H{o}s-Renyi Random Graph to simulate the network topology, where links are established with probability 0.6 and the maximum node degree is bounded by 8.
The traffic demand over each link is uniformly distributed in the range [1,100]Mbps. The capacity of each white channel is uniformly distributed in the range [75,200]Mbps. In our simulation, 5000 random graph instances are tested.
\subsection{Scaling of Recovery Capacity}\label{sec::simulation-scaling}
We first investigate how the relative recovery capacity ratio (see Section \ref{sec::scaling} for the definition) scales with the network size. As is observed in Fig. \ref{fig:scaling}, the recovery capacity ratio goes down with the growth of the network size. Specifically, the required recovery capacity is only around 1\% of the total traffic demands in a 200-node network, even with very few white channels (e.g., $|W|=3$). Therefore, we expect the recovery capacity to become negligible as compared to the total traffic demands as the network size continues to grow. In addition, curve fitting shows that the recovery capacity ratio scales as $\Theta(\frac{1}{|V|^a})$ where $a$ ranges in 1.02-1.09, which roughly matches the theoretical bound we obtain in Theorem \ref{thm::capacity-ratio}. 
\begin{figure}[ht]
\begin{center}
\includegraphics[width=2.5in]{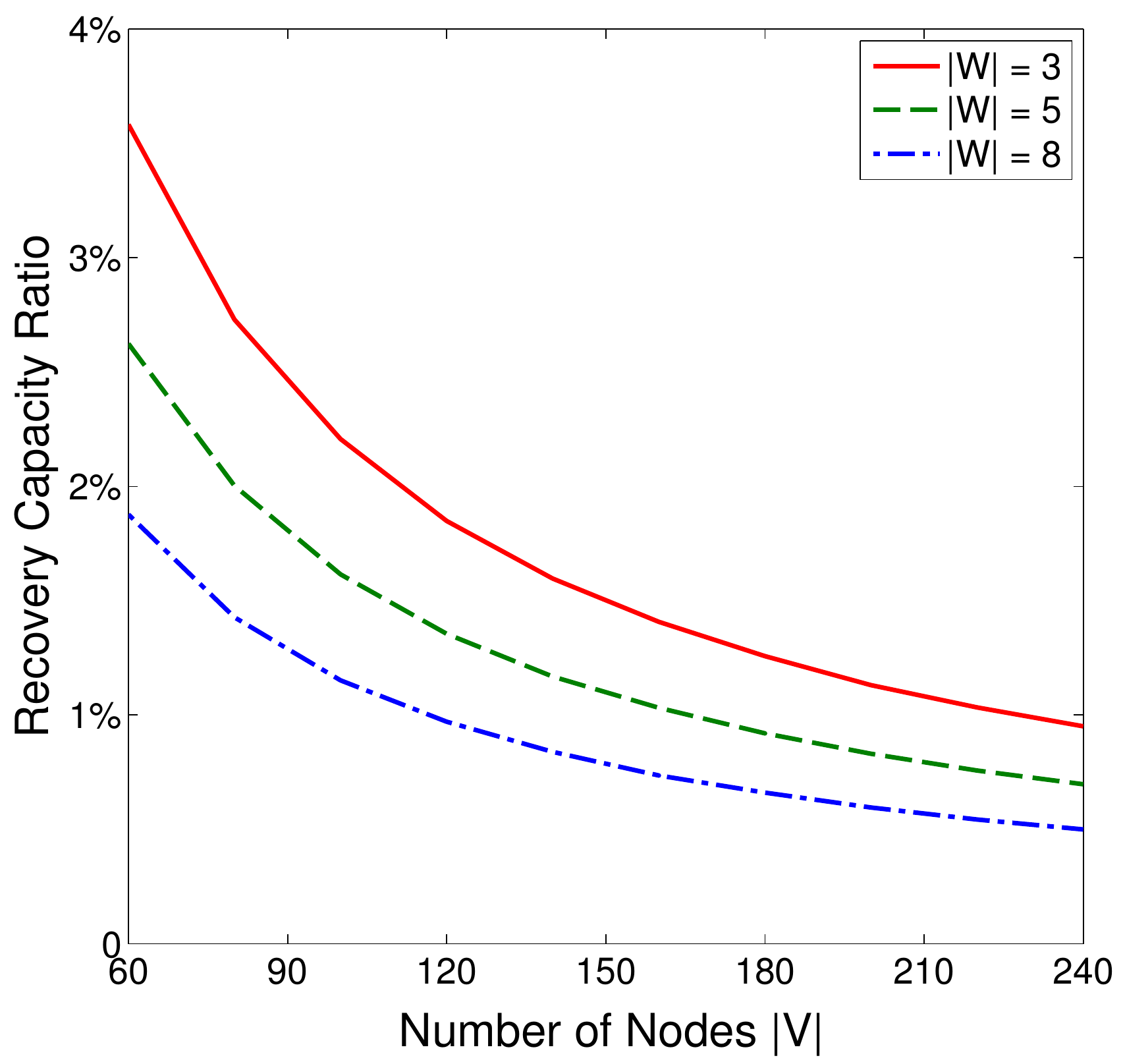}\vspace{-0.3cm}
\caption{Scaling of the relative recovery capacity ratio with the network size $|V|$ (where $k=2$ channel preemptions are to be survived).}
\label{fig:scaling}\vspace{-3mm}
\end{center}
\end{figure}

\subsection{Approximation Quality}
Since we consider the bi-criteria approximation framework, two metrics should be evaluated: the recovery capacity and the fraction of traffic sustained over each link. Through the rest of this section, we study the two aspects by comparing the following schemes.
\begin{itemize}
\item Greedy Algorithm (Algorithm \ref{alg::greedy}).
\item Interference-Free Assignment (IFA). Note that an interference-free assignment is guaranteed to exist only if $|W|>d_{\max}$ (in our simulation, $d_{\max}=8$).
\item Random Assignment (RndAssign) that assigns each link a random white channel.
\item Optimal result to $\mathsf{WhiteRec}$, computed with Gurobi, a large-scale mathematical programming solver.
\end{itemize}

\noindent \textbf{Recovery Capacity.}
Fig. \ref{fig:recovery} illustrates the comparison of these schemes with respect to the recovery capacity; Table \ref{t1} lists the detailed approximation gap\footnote{The approximation gap is defined by $\frac{\mathsf{ALG}-\mathsf{OPT}}{\mathsf{OPT}}$, where $\mathsf{ALG}$ is the amount of required recovery capacity by using the approximation algorithm and $\mathsf{OPT}$ is the minimum recovery capacity.}.
We first focus on the approximation quality of IFA. When $k=1$, IFA yields the same amount of recovery capacity as the optimal solution and the approximation gap is zero. In fact, it can be analytically shown that IFA is optimal when $k=1$. When $k=2$, IFA is only slightly worse than the optimum (less than 2\%, as is shown in Table \ref{t1}), much better than the 1.25-approximation bound. The only caveat is that IFA is guaranteed to exist only if $|W|>d_{\max}$.

Next, we investigate the approximation quality of the greedy algorithm. Despite its relatively loose approximation ratio, the greedy algorithm performs very well in practice. The worst approximation gap is 26\% when $k=1$ and 14\% when $k=2$. It also outperforms the random assignment by almost an order of magnitude in terms of the approximation gap. When compared with IFA, the greedy algorithm is slightly worse but it has the  advantage of being applicable in any scenario.

\begin{figure}[ht]
\begin{center}
\includegraphics[width=3.0in]{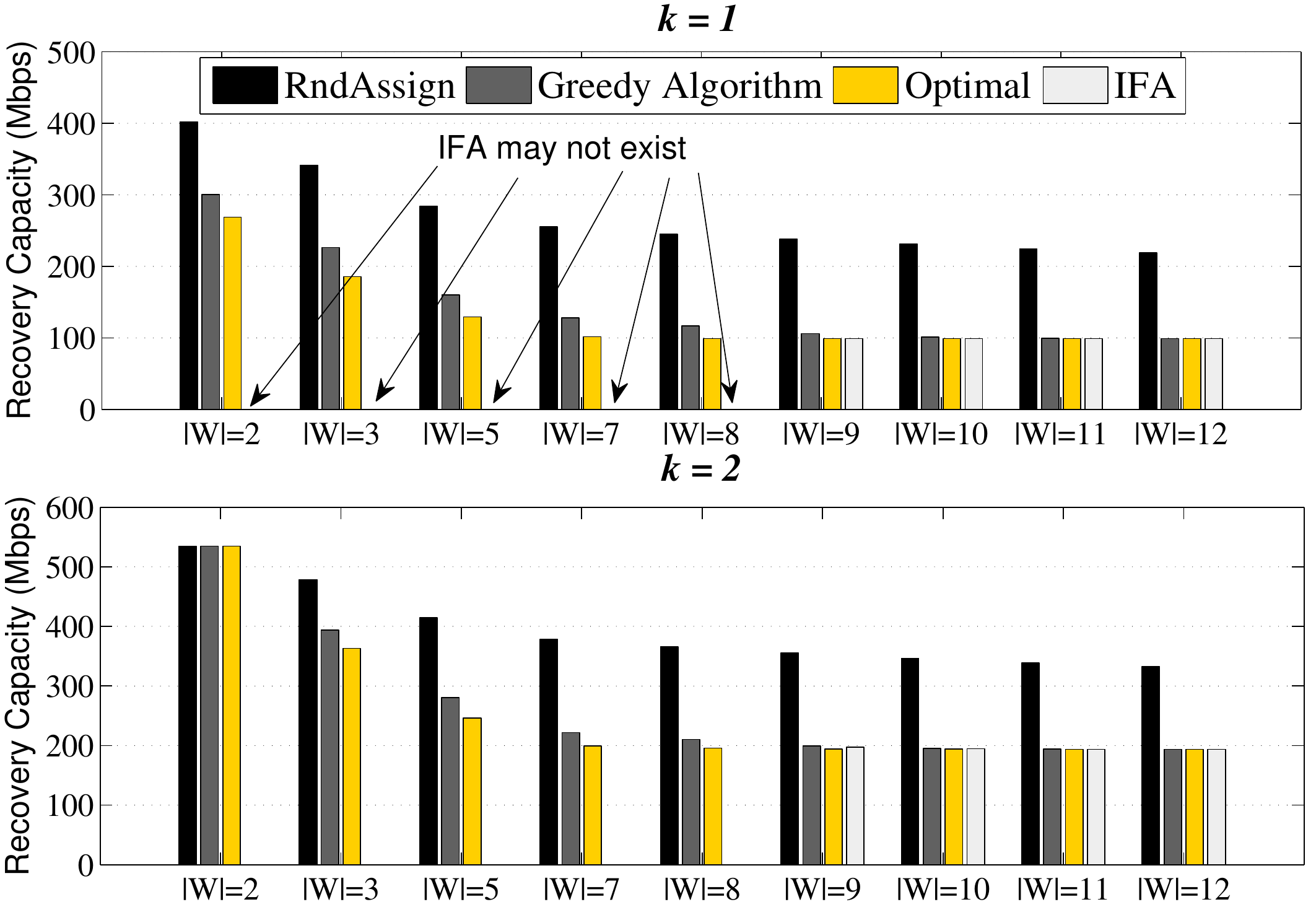}\vspace{-0.3cm}
\caption{The comparison among different algorithms with respect to the required recovery capacity ($|V|=20$).}
\label{fig:recovery}
\end{center}\vspace{-3mm}
\end{figure}


\begin{table}[h]
\caption{\label{t1}Approximation gap of different schemes}
\centering
\small
\begin{tabular}{|c|c|c|c|l|c|c|c|}
\hline
\multicolumn{4}{|c|}{\textbf{Survive $k=1$ failures}} &  & \multicolumn{3}{c|}{\textbf{Survive $k=2$ failures}} \\ \hline
$|W|$        & Rnd         & Greedy       & IFA       &  & Rnd             & Greedy            & IFA            \\ \hline
2            & 50\%        & 12\%         & N/A       &  & 0\%             & 0\%               & N/A            \\ \hline
3            & 84\%        & 22\%         & N/A       &  & 32\%            & 9\%               & N/A            \\ \hline
5            & 120\%       & 24\%         & N/A       &  & 68\%            & 14\%              & N/A            \\ \hline
7            & 151\%       & 26\%         & N/A       &  & 90\%            & 11\%              & N/A            \\ \hline
8            & 148\%       & 18\%         & N/A       &  & 87\%            & 7\%               & N/A            \\ \hline
9            & 140\%       & 7\%          & 0\%       &  & 83\%            & 3\%               & 2\%            \\ \hline
10           & 134\%       & 3\%          & 0\%       &  & 79\%            & 1\%               & 0\%            \\ \hline
11           & 127\%       & 1\%          & 0\%       &  & 75\%            & 0\%               & 0\%            \\ \hline
12           & 122\%       & 0\%          & 0\%       &  & 72\%            & 0\%               & 0\%            \\ \hline
\end{tabular}
\end{table}

\noindent \textbf{Sustainable Traffic.}
In Fig. \ref{fig:loads}, we illustrate the comparison among different assignment schemes in terms of the fraction of traffic sustained over each link. Note that the maximum sustainable traffic level is obtained by solving $\mathsf{FEASI}$ (see Section \ref{subsec::feasibility}) in Gurobi.  We first notice that if there is only a small number of white channels, the maximum sustainable traffic level can be less than 100\%.  With more white channels, we have more spectrum resources and 100\% traffic demands are sustainable. By comparison, the fraction of traffic sustained by the greedy algorithm is reasonably good as compared to the maximum sustainable level (at least 60\% of the maximum), and the greedy algorithm significantly outperforms the random assignment. In particular, given a sufficient number of white channels (say $|W|\ge 9$), the greedy algorithm yields a similar performance to IFA and sustains over 90\% traffic demands. 

\begin{figure}[ht]
\begin{center}
\includegraphics[width=2.6in]{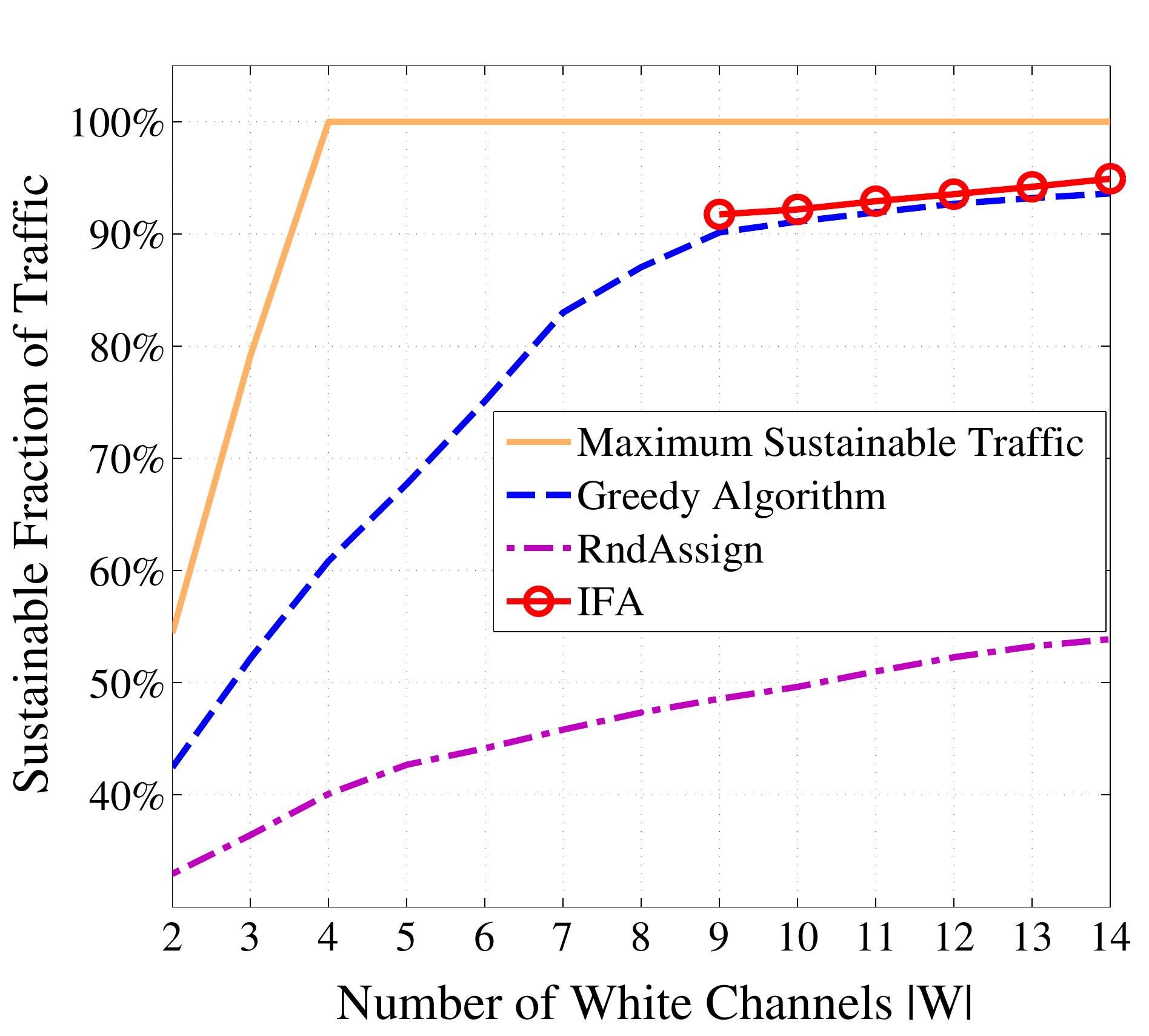}
\caption{Comparison among different algorithms with respect to the fraction of traffic  sustained over each link ($|V|=20$).}
\label{fig:loads}\vspace{-0.3cm}
\end{center}
\end{figure}

\section{Conclusions}\label{sec::conclusion}
In this paper, we use backup channels to provide reliability guarantees for secondary users. In particular, we investigate the optimal white channel assignment that minimizes the recovery capacity required to survive a given number of white channel preemptions. This problem is shown to be coupled by two NP-hard objectives, and two bi-criteria approximation schemes are developed. Moreover, we show that the required recovery capacity is negligible as compared to the total network traffic in a large-scale network, which demonstrates the scalability of this framework.

\appendix

\section{Proofs}
\subsection{Proof to Lemma \ref{thm::bipartite-C(y,k)-bound}} \label{proof::thm::bipartite-C(y,k)-bound}
Let
\[
(U^*,S^*)=\max_{U\in\mathcal{V},S\in\mathcal{W}(k)}\frac{1}{\alpha(U)}\sum_{w\in S}\sum_{e\in E(U)}r_ey_e^w.
\]
The induced graph on $U^*$ is denoted by $G^*=(U^*,E(U^*))$, which is still a bipartite graph. Let $C$ be the minimum vertex cover of $G^*$. Then we have
\[
\begin{split}
M_2(y,k)&=\frac{1}{\alpha(U^*)}\sum_{e\in E(U^*)}\sum_{w\in S^*} r_ey_e^w\\
&\le \frac{1}{\alpha(U^*)}\sum_{v\in C}\sum_{e\in \delta(v)\cap E(U^*)}\sum_{w\in S^*}r_ey_e^w\\
&\le \frac{1}{\alpha(U^*)} \sum_{v\in C}\sum_{e\in\delta(v)}\sum_{w\in S^*}r_ey_e^w\\
&\le \frac{1}{\alpha(U^*)} \sum_{v\in C}M_1(y,k)\\
&=\frac{1}{\alpha(U^*)} |C|M_1(y,k),
\end{split}
\]
where the first inequality holds because every edge in $E(U^*)$ is incident on at least one node in $C$. By K$\ddot{\text{o}}$nig's Theorem \cite[pp. 203–207]{Konig},  the size of the minimum vertex cover equals to the size of the maximum matching in a bipartite graph. Thus,  $|C|$ is upper-bounded by $\lfloor \frac{|U^*|}{2} \rfloor= \frac{|U^*|-1}{2}$ (note that $|U^*|$ is odd). Therefore, we can finally conclude that
\[
M_2(y,k)\le \frac{1}{\alpha(U^*)} |C|M_1(y,k)=M_1(y,k).
\]
\subsection{Proof to Theorem \ref{thm::greedy-approx}} \label{proof::thm::greedy-approx}
Before the detailed proof, we first introduce a relaxed problem called $\mathsf{WhiteRecInf}$, which is the same as $\mathsf{WhiteRec}$ but assumes infinite channel capacity such that any channel assignment can support the given traffic demands. In other words, feasibility conditions \eqref{eqn:node-constraint-y} and \eqref{eqn:oddset-constraint-y} are relaxed.
\begin{align}
\mathsf{WhiteRecInf}:\quad \min_{y}& \quad C(y,k)\nonumber\\
\textrm{ s.t.}&\sum_{w\in W}y_e^w=1,\forall e\in E\nonumber\\
&y_e^w\in\{0,1\},\forall e\in E, w\in W\nonumber.
\end{align}

We first show that the greedy algorithm yields no more than $\rho$ times of the minimum recovery capacity in $\mathsf{WhiteRec}$. The following notations are necessary.\\
$\bullet$ $y_{\mathsf{OPT}}$ is the optimal solution to $\mathsf{WhiteRec}$.\\
$\bullet$ $y^*$ is the optimal solution to $\mathsf{WhiteRecInf}$.\\
$\bullet$ $\hat{y}$ is the solution given by the greedy algorithm. Also denote $\hat{E}_w$ the set of links that are assigned channel $w$ under $\hat{y}$.

It is clear that $C(y^*,k)\le C(y_{\mathsf{OPT}},k)$ since $\mathsf{WhiteRecInf}$ is the relaxed problem of $\mathsf{WhiteRec}$. Hence, it suffices to prove $C(\hat{y},k)\le \rho C(y^*,k)$. To facilitate our proof, we introduce a lemma.

\begin{lemma}\label{lm::greedy-approx}
 $C(y^*,k)\ge \frac{k}{|W|}\max_{v\in V}\sum_{e\in \delta(v)}r_e$.
\end{lemma}
\begin{proof}
Let $v_1$ be the node with the maximum traffic demands, i.e., $v_1=\arg\max_{v\in V}\sum_{e\in \delta(v)}r_e$. Also denote $L^*_w$ the total traffic supported by white channel $w$ at node $v_1$ under assignment $y^*$. Without loss of generality, we assume  $L^*_{1}\ge L^*_{2}\ge...\ge L^*_{|W|}$.
Then it follows that
\[
\begin{split}
M_1(y^*,k)&=\sum_{1\le w\le k} L^*_w.
\end{split}
\]
If $M_1(y^*,k) < \frac{k}{W}\sum_{e\in \delta(v_1)}r_e$, we would obtain
\[
\frac{\sum_{1\le w\le k} L^*_w}{k}< \frac{1}{W}\sum_{e\in \delta(v_1)}r_e,
\]
i.e., the average traffic (at node $v_1$) in the first $k$ white channels are smaller than the average traffic (at node $v_1$) in all the white channels. This is an obvious contradiction since the first $k$ channels support more loads at node $v_1$ than the remaining $|W|-k$ channels. Hence we can conclude that
\[
C(y^*,k) \ge M_1(y^*,k)\ge  \frac{k}{W}\sum_{e\in \delta(v_1)}r_e=\frac{k}{W}\max_{v\in V}\sum_{e\in \delta(v)}r_e.
\]
This completes the proof to the lemma.
\end{proof}

\noindent Now we get back to proving $C(\hat{y},k)\le \rho C(y^*,k)$. Define
\[
v^*=\arg\max_{v\in V, S\in\mathcal{W}(k)}\sum_{w\in S}\sum_{e\in \delta(v)\cap \hat{E}_w}r_e,
\]
and without loss of generality, we suppose
\[
\sum_{e\in \delta(v^*)\cap \hat{E}_1}r_e \ge \sum_{e\in \delta(v^*)\cap \hat{E}_2}r_e \ge \cdots \sum_{e\in \delta(v^*)\cap \hat{E}_{|W|}}r_e.
\]
Note that under the above definitions, we have
\[
\begin{split}
M_1(\hat{y},k)&=\sum_{1\le w\le k}\sum_{e\in \delta(v^*)\cap \hat{E}_w}r_e.
\end{split}
\]
Suppose $e_w=(v^*, u_w^*)$ is the last edge added to $\hat{E}_w$ that is incident on $v^*$, and denote $D_w$ the set of edges that have been assigned a white channel before edge $e_w$. Then it follows that for any white channel $w\in W$
\begin{align}
\sum_{e\in \delta(v^*)\cap \hat{E}_w}r_e &= \sum_{e\in \delta(v^*)\cap \hat{E}_w \cap D_w}r_e+r_{e_w} \label{app_eq1}\\
&\le\sum_{e\in \delta(v^*,u_w^*)\cap \hat{E}_w \cap D_w}r_e+r_{e_w}\label{app_eq2}\\
&\le \frac{1}{|W|} \sum_{e\in \delta(v^*,u_w^*)\cap D_w}r_e+r_{e_w}\label{app_eq3}\\
&\le \frac{1}{|W|}  \Big(\sum_{e\in \delta(v^*)}r_e+ \sum_{e\in \delta(u_w^*)}r_e\Big)+\frac{|W|-2}{|W|}r_{e_w}\label{app_eq4}\\
&\le \frac{2}{|W|} \max_{v\in V}\sum_{e\in \delta(v)}r_e+\frac{|W|-2}{|W|}r_{e_w}\label{app_eq5}.
\end{align}
Here, \eqref{app_eq1} holds because edge $e_w$ is the last one added to $\hat{E}_w$ that is incident on $v^*$; \eqref{app_eq3} is due to the fact edge $e_w$ is assigned channel $w$ only if channel $w$ has the minimum aggregate loads at node $v^*$ and $u_w^*$ among all white channels (see step \ref{s4} in the greedy algorithm); \eqref{app_eq4} holds because $e_w$ is incident on both $v^*$ and $u_w^*$ while $D_w$ excludes $e_w$.
Then we have
\[
M_1(\hat{y},k)\le \frac{2k}{W}\max_{v\in V}\sum_{e\in\delta(v)}r_e+\frac{|W|-2}{|W|}\sum_{1\le w\le k}r_{e_w}.
\]
By Lemma \ref{lm::greedy-approx}, we know
\[
\frac{k}{|W|}\max_{v\in V}\sum_{e\in\delta(v)}r_e\le C(y^*,k).
\]
At the same time, notice that $e_1,e_2,...,e_{|W|}$ are distinct edges incident on $v^*$. Then it is easy to see that
\[
\sum_{1\le w\le k}r_{e_w}\le C(y^*,k).
\]
Therefore, we can conclude that
\[
M_1(\hat{y},k)\le (3-\frac{2}{|W|}) C(y^*,k) = \frac{2}{3}\rho C(y^*,k).
\]
By Theorem \ref{thm:C(y,k)-bounds}, we finally have
\[
C(\hat{y},k) \le \frac{3}{2} M_1(\hat{y},k)\le \rho C(y^*,k).
\]

We now show that  at least $\frac{1}{\rho}\frac{R_{\min}}{R_{\max}}$-fraction of traffic can be sustained by the greedy assignment.
Let $\hat{y}$ be the solution obtained by the greedy algorithm, and denote by $\tilde{y}$ the optimal solution to $\mathsf{FEASI}$. Then it follows from Lemma \ref{lm::feasibility} that
\begin{equation}\label{bi_eq1}
\frac{\beta(\hat{y})}{\beta(\tilde{y})}\ge \frac{C(\tilde{y},1)}{C(\hat{y},1)}\frac{R_{\min}}{R_{\max}}.
\end{equation}
Denote $\mathsf{OPT}_1$ the minimum recovery capacity required to survive \textbf{one preemption} in $\mathsf{WhiteRec}$. Note that $\hat{y}$ is intended for surviving any $k$ preemptions. However, the greedy algorithm is invariant to the number of preemptions we need to survive so $\hat{y}$ is also the greedy assignment for surviving one preemption. Thus, we have $C(\hat{y},1)\le \rho \mathsf{OPT}_1$ according to the first approximation ratio. Also note that $C(\tilde{y},1)\ge \mathsf{OPT}_1$. Then
\begin{equation}\label{bi_eq2}
\frac{C(\tilde{y},1)}{C(\hat{y},1)}\ge \frac{\mathsf{OPT}_1}{\rho \mathsf{OPT}_1}=\frac{1}{\rho}.
\end{equation}
Taking \eqref{bi_eq2} into \eqref{bi_eq1}, we have
\[
\frac{\beta(\hat{y})}{\beta(\tilde{y})}\ge \frac{1}{\rho}\frac{R_{\min}}{R_{\max}} .
\]
Since $\mathsf{WhiteRec}$ has a feasible solution, then $\beta(\tilde{y})\ge 1$ and $\beta(\hat{y}) \ge \frac{1}{\rho }\frac{R_{\min}}{R_{\max}}$. This completes our proof.
\subsection{Proof to Corollary \ref{thm::ifca-optimal}} \label{proof::thm::ifca-optimal}
We only prove that any IFA yields the minimum recovery capacity in above three scenarios while the ratio for the sustainable traffic follows the same argument as in the proof to Theorem \ref{thm::ifca-bound}.

\noindent \textbf{Part (i).} Since $y$ is interference-free and $k=1$, we have
\[
M_1(y,k)=\max_{v\in V,w\in W}\sum_{e\in\delta(v)}r_ey_e^w=r_{\max},
\]
\vspace{-2mm}
\begin{align}
M_2(y,k)&\leq
r_{\max}\max_{U\in\mathcal{V},w\in W}\frac{2}{|U|-1}\sum_{e\in E(U)}y_e^w\nonumber\\
&\leq r_{\max}\max_{U\in\mathcal{V},w\in W}\frac{2}{|U|-1}\left\lfloor\frac{|U|}{2}\right\rfloor\label{eqn::optimal1}\\
&= r_{\max}\nonumber,
\end{align}
where \eqref{eqn::optimal1} is due to the fact that $y$ is interference-free and consequently, the number of links in $U\subseteq V$ using the same channel is upper-bounded by the size of maximum matching in $E(U)$, which is $\left\lfloor\frac{|U|}{2}\right\rfloor$. The above two bounds show that $C(y,1)=r_{\max}$. Since $C(\tilde{y},1)\geq r_{\max}$ for any assignment $\tilde{y}$, we can conclude that $y$ yields the minimum recovery capacity.

\vspace{1mm}

\noindent \textbf{Part (ii).} Consider an arbitrary interference-free channel assignment $y$. Since $y$ is interference-free, all the links incident to a node are assigned different channels. Furthermore, we have $k\leq d_{\max}$. Consequently, there exists a node $v$ and $S\in\mathcal{W}(k)$ such that $\sum_{w\in S}\sum_{e\in\delta(v)}y_e^w=k$. Using this observation, the value $M_1(y,k)$ can be rewritten as
\begin{align*}
M_1(y,k)=&r\max_{v\in V,S\in\mathcal{W}(k)}\sum_{w\in S}\sum_{e\in\delta(v)}y_e^w=rk
\end{align*}

On the other hand, using the similar trick to Part (i), we can bound $M_2(y,k)$ by $M_2(y,k)=rk$.
Hence, we have $C(y,k)=M_1(y,k)=rk$.
Clearly, $C(\tilde{y},k)\ge rk$ any assignment $\tilde{y}$; thus, $y$ yields the minimum recovery capacity.

\vspace{1mm}

\noindent \textbf{Part (iii).} This part directly follows from Lemma \ref{thm::bipartite-C(y,k)-bound} and Lemma \ref{lem:ifca-minimize-m1}.

\begin{thebibliography}{1}
\small

\bibitem{Mitola-CR} J. Mitola III, ``Cognitive radio for flexible mobile multimedia communications," \emph{ACM/Kluwer MONET}, vol. 6, no. 5, pp. 435–441, Sep. 2001.

\bibitem{white-database} FCC. Order, FCC 11-131. 2011.

\bibitem{Shin-DySPAN-08} H. Kim and K. Shin, ``Fast discovery of spectrum opportunities in
cognitive radio networks," \emph{IEEE DySPAN}, 2008.

\bibitem{Liang-JSAC-12} Q. Liang, S. Han, F. Yang, G. Sun, and X. Wang, ``A Distributed-Centralized Scheme for Short- and Long-Term Spectrum Sharing with a Random Leader in Cognitive Radio Networks," \emph{IEEE J. Sel. Areas Commun.}, vol. 30, no. 11, pp. 2274-2284, 2012.

\bibitem{Yue-ICC-08} K. F. Li, W. C. Lau, and O. C. Yue, ``Link restoration in cognitive radio networks," \emph{IEEE ICC}, 2008.


\bibitem{Chowdhury-JSAC-08} K. Chowdhury and I. Akyildiz, ``Cognitive wireless mesh networks with dynamic spectrum access," \emph{IEEE J. Sel. Areas Commun.}, vol. 26, no. 1, pp. 168–181, 2008.

\bibitem{Chowdhury-adhoc-09} I. Akyildiz, W. Lee, and K. Chowdhury, ``CRAHNs: Cognitive radio ad hoc networks," \emph{Ad Hoc Networks}, vol. 7, no. 5, pp. 810–836, 2009.

\bibitem{Liang-TON-14} Q. Liang, X. Wang, X. Tian, F. Wu, and Q. Zhang, ``Two-Dimensional Route Switching in Cognitive Radio Networks: A Game-Theoretical Framework," \emph{IEEE/ACM Transactions on Networking}, vol. 23, no. 4, pp. 1053-1066, 2015.

\bibitem{Chung-ICC-13} P. Tseng and W. Chung, ``Local Rerouting and Channel Recovery for Robust Multi-Hop Cognitive Radio Networks," \emph{IEEE ICC}, 2013.

\bibitem{Zhao-JSAC-07} Q. Zhao, L. Tong, A. Swami, and Y. Chen, ``Decentralized cognitive mac for opportunistic spectrum access in ad hoc networks: A POMDP
framework," \emph{IEEE J. Sel. Areas Commun.}, vol. 25, no. 3, pp. 589–600, 2007.

\bibitem{Kuo-WCNC-07} A. Chia-Chun Hsu, D. Weit, and C.-C. Kuo, ``A cognitive MAC protocol using statistical channel allocation for wireless ad-hoc networks," \emph{IEEE WCNC}, 2007.

\bibitem{Cao-INFOCOM-12} J. Zhao and G. Cao, ``Robust topology control in multi-hop cognitive radio networks," \emph{IEEE INFOCOM}, 2012.

\bibitem{Mihnea-ICNC-14} M. Cardei and A. Mihnea, ``Channel Assignment in Cognitive Wireless Sensor Networks," \emph{IEEE ICNC}, 2014.




\bibitem{Modiano-Mobihoc-12} K. Jagannathan, I. Menache, E. Modiano, and G. Zussman, ``Noncooperative spectrum access - the dedicated vs. free spectrum choice," \emph{ACM Mobihoc}, 2011.

\bibitem{one-hop1} E. Modiano, D. Shah, and G. Zussman, ``Maximizing Throughput in Wireless Networks via Gossiping," \emph{ACM SIGMETRICS}, 2006.

\bibitem{one-hop2} L. Chen, S. H. Low, M. Chiang, and J. C. Doyle, ``Optimal cross-layer congestion control, routing and scheduling design in
ad hoc wireless networks," \emph{IEEE INFOCOM}, 2006.

\bibitem{mmWave} Sumit Singh, Raghuraman Mudumbai, and Upamanyu Madhow, ``Distributed coordination with deaf neighbors: efficient medium access for 60 GHz mesh networks," \emph{IEEE INFOCOM,} 2010.

\bibitem{Hajek-IT-88} B. Hajek and G. Sasaki, ``Link scheduling in polynomial time," \emph{IEEE Transactions on Information Theory}, vol. 34, no. 5, pp. 910 –917,
1988.

\bibitem{garey:johnson} M. Garey and D. Johnson. \emph{Computers and Intractability: A Guide to the Theory of NP-Completeness}. W. H. Freeman \& Co., 1990.

\bibitem{bi-criteria} M. Ehrgott. \emph{Multicriteria Optimization.} Lecture Notes in Economics and Mathematical Systems, Springer, 2000.

\bibitem{graham} R. L. Graham. Bounds on Multiprocessing Timing Anomalies. \emph{SIAM J. On Applied Math.}, vol. 17, no. 2, pp. 416-429, 1969.

\bibitem{coloring-alg} J. Misra, and D. Gries, ``A constructive proof of Vizing's Theorem,"  \emph{Information Processing Letters}, vol. 41, no. 3, pp. 131–133, 1992.


\bibitem{vizing:chromatic} V. Vizing. On an estimate of the chromatic class of a p-graph. \emph{Diskret. Analiz.}, vol. 3, pp. 25–30, 1964.

\bibitem{Edmond} J. Edmonds, ``Maximum matching and a polyhedron with (0,1) vertices,"  \emph{J. Res. Nat. Bur. Standards}, vol. 69, pp. 125–130, 1965.

\bibitem{WDM} S. Subramaniam and R.Barry, ``Wavelength assignment in fixed routing WDM networks," \emph{IEEE ICC}, 1997.

\bibitem{Konig} E. K. Biggs, Lloyd, and R. J. Wilson. \emph{Graph Theory 1736–1936}. Oxford University Press, 1976.

\end{thebibliography}
\end{document}